\documentclass{amsart}

\usepackage[dvipsnames]{xcolor}
\usepackage{fullpage,graphicx,subfigure,mathpazo,color}
\usepackage{amsmath,amscd,tikz,mathrsfs}
\usepackage[normalem]{ulem}
\usepackage{amsmath}
\usepackage{epstopdf}
\usepackage{tikz}
\usepackage{setspace}

\newcommand{\ii}{\mathrm{i}}
\newcommand{\ee}{\mathrm{e}}

\newcommand{\T}{\mathrm{T}}

\newtheorem{rhp}{Riemann-Hilbert Problem}
\newtheorem{theorem}{Theorem}
\newtheorem{lemma}{Lemma}
\newtheorem{prop}{Proposition}
\newtheorem{corollary}{Corollary}

\newtheorem{remark}{Remark}

\usepackage{graphicx}      
\usepackage{titlesec}

\titleformat{\section}{\centering\LARGE\bfseries}{\thesection}{1em}{}
\titleformat{\subsection}{\Large\bfseries}{\thesubsection}{1em}{}

\begin{document}

\title{Rational solutions of Painlev\'{e}-II equation as Gram determinant}

\author{Liming Ling}
\address{School of Mathematics, South China University of Technology, Guangzhou, China 510641}
\email{linglm@scut.edu.cn}
\author{Bing-Ying Lu}
\address{Department of Mathematics, University of Bremen, Bremen 28359, Germany}
\email{lubi@uni-bremen.de}
\author{Xiaoen Zhang}
\address{School of Mathematics, South China University of Technology, Guangzhou, China 510641}
\email{zhangxiaoen@scut.edu.cn}

\begin{abstract}
Under the Flaschka-Newell Lax pair, the Darboux transformation for the Painlev\'{e}-II equation is constructed by the limiting technique. With the aid of the Darboux transformation, the rational solutions are represented by the Gram determinant, and then we give the large $y$ asymptotics of the determinant and the rational solutions. Finally, the solution of the corresponding Riemann-Hilbert problem is obtained from the Darboux matrices.


\end{abstract}

\date{\today}

\maketitle
\section{Introduction}

The six Painlev\'e equations ($P_{I}-P_{IV}$) are a class of nonlinear ODE which have long been studied. It originated from when the Painlev\'e school tried to answer Picard's question in \cite{Picard-Jde-1889} in the late 19th century. Because the general solutions cannot be reduced to elementary or known functions, today it is often regarded as a nonlinear type of special function. Due to transcendency, the solutions to the Painlev\'e equations are often called the Painlev\'e transcendents. In the late 1970s, Painlev\'e equations garnered new attention when they are found to be ``integrable'', thus not only extending the tools to analyze their solutions, but also expanding the theory of integrability. These equations also play central role in mathematical physics. For instance, they appear in high-energy physics \cite{Johnson-arxiv-2006}, in planar fluid flow \cite{Clarkson-study-2009}.

We are particularly interested in certain pattern formation related to rogue waves in integrable nonlinear waves and their connection to the Painlev\'e transcendents. Very recently, the locations of the rogue wave patterns in nonlinear Schr\"odinger equation and the lump patterns in Kadomtsev-Petviashvili equation are found to be determined by the root structure of the Yablonskii-Vorob'ev polynomials hierarchy \cite{Dong-phyd-2022,Yang-phyd-2021,Yang-phyd-2021-1}. It has long been observed in \cite{Yablonskii-1959} that the Painlev\'e-II rational solutions can be represented by the Yablonskii-Vorob'ev polynomials, so the root structure is exactly the pole structure for Painlev\'e-II rational solutions. This directly inspired us to add to the ample existing studies of Painlev\'e-II rational solutions from a different perspective.

The general formula of Painlev\'{e}-II equation is given by
\begin{equation}\label{eq:P-II}
p_{yy}-2p^3-\frac{2}{3}py+\frac{2}{3}m=0.
\end{equation}
As mentioned, the general solutions of Painlev\'e equations are transcendental. However, Painlev\'e-II equation is known to have rational solutions for some special parameters. Indeed, in \cite{Yablonskii-1959} Yablonskii first found an iteration formula for the rational solutions with $m\in\mathbb{Z}$ in \eqref{eq:P-II}, and Airault \cite{Airault-Study-1979} showed that the parameter choice is both necessary and sufficient. Like soliton solutions in integrable systems, the rational solutions can be given by B\"acklund transformations. The B\"acklund transformation for Pailev\'e-II equation is given by Lukashevich in $1971$ \cite{Luka-DE-1971}. In $1985$, \cite{Murata-FE-1985} derived the rational solutions of the second and the fourth Painlev\'{e} equation with different integer parameters, and in particular showed the necessary and sufficent conditions for rational solutions to exist. There are numerous study of rational solutions \cite{Airault-Study-1979,Clarkson-JCAM-2003,Clarkson-Study-2020,Clarkson-nonl-2003,Fokas-JMP-1981,Ohta-JMP-1996, Okamoto-Math-1986} through the B\"{a}cklund transformation.

Painlev\'e equations are considered integrable in the sense that they admit Lax representation. Indeed, Painlev\'e-II equation admits the following Flaschka-Newell Lax representation
\begin{equation}\label{eq:Lax-pair}
\begin{aligned}
&\frac{\partial \pmb{v}}{\partial\lambda}=\pmb{A}^{\rm FN}(\lambda; y)\pmb{v}, \quad \pmb{A}^{\rm FN}(\lambda; y):=\begin{bmatrix}-6\ii\lambda^2-3\ii p^2-\ii y&6p\lambda+3\ii p'+m\lambda^{-1}\\6p\lambda-3\ii p'+m\lambda^{-1}&6\ii\lambda^2+3\ii p^2+\ii y
\end{bmatrix},\\
&\frac{\partial \pmb{v}}{\partial y}=\pmb{B}^{\rm FN}(\lambda; y)\pmb{v}, \quad \pmb{B}^{\rm FN}(\lambda; y):=\begin{bmatrix}-\ii\lambda&p\\p&\ii\lambda
\end{bmatrix}.
\end{aligned}
\end{equation}
The compatibility condition
\begin{equation*}
\frac{\partial \pmb{A}^{\rm FN}}{\partial y}-\frac{\partial \pmb{B}^{\rm FN}}{\partial\lambda}+\left[\pmb{A}^{\rm FN},\pmb{B}^{\rm FN}\right]=0
\end{equation*}
yields the Painlev\'{e}-II equation \eqref{eq:P-II}. Due to integrability, one powerful tool for analysis of the Painlev\'e equations is the Riemann-Hilbert problems. The book \cite{Fokas-Pain-book} is a particularly comprehensive and systematic documentation of this modern subject.  The Riemann-Hilbert problem is also used to study the rogue wave patterns that we are interested in. Specifically in literature, \cite{tovbis-CPAM-2013} found that in semiclassical focusing nonlinear Schr\"odinger equation, near a gradient catastrophe point, universal ``rogue wave''-like patterns form can be mapped to the poles of the Painlev\'e-I tritronqu\'ee solution. In \cite{Bilman-Duke-2019}, the authors analyzed the asymptotics of the infinite order rogue waves, and the asymptotics in the transitional region was related to the Tritronqu\'{e}e solutions of Painlev\'{e}-II equation.

There are many ways to derive Painlev\'e-II rational solutions. One standard method is the B\"{a}cklund transformation, given nearly $50$ years ago. However, to the best of our knowledge, few has studied the rational solution of Painlev\'e equations with the Darboux transformation directly. The Darboux transformation for Painlev\'e-II rational solutios is still an open problem up to now.  Observe that the $\partial y$-equation in the Lax pair \eqref{eq:Lax-pair} is similar to the spectral problem in the AKNS system. It is also shown that the Painlev\'{e}-II equation is related the modified Korteweg-de Vries equation via a self-similar transformation \cite{Ablowitz-LNC-1978,Ablowitz-PRL-1977,Flaschka-CMP-1980,Kitaev-Math-1985}. Thus it is natural to borrow the idea from the Darboux transformation of AKNS system to construct the solution of the Painlev\'{e}-II equation.

In the present work, we modify the Darboux transformation to construct the rational Painlev\'e-II solutions. Under the same framework, we also derive the rational solution with a Gram determinant representation. With our determinantal formula, we get the asymptotics of large $y$ by expansion. The rational solutions of Painlev\'e-II equation are often compared to the soliton solutions in integrable nonlinear wave equations, since they can be represented by determinantal formula, and correspond to fully discrete spectral data etc.\mbox{ }\cite{Flaschka-CMP-1980}. However, the Flaschka-Newell determinantal formula is different with our Gram determinant formula. Indeed, in our Darboux transformation, all of the spectral data are at the same point so we have to employ limiting technique. The process is much more like deriving the rogue wave solutions in integrable nonlinear waves.
The Gram determinant representation is also comparable to rogue waves. This also in a way attests for the connection between the roots of Yablonskii-Vorob'ev polynomial (poles of Painlev\'e-II rational solutions) and rogue wave patterns found in \cite{Yang-phyd-2021,Yang-phyd-2021-1}. However, the transformation for Painlev\'e-II has even one more difference, the generalized Darboux-B\"{a}cklund transformation we used is not the so called auto-transformation in comparison, which we believe also makes our work even more interesting.

Ultimately we believe to study the rogue wave pattern, Riemann-Hilbert method can give us a better analytical handle, therefore we always have the Riemann-Hilbert problem in mind. In \cite{Bilman-Duke-2019}, the authors studied the asymptotics of infinite order rogue waves via the Riemann-Hilbert method, which was first constructed by using the Darboux transformation. Based on this idea, we are able to solve the Painlev\'e-II Riemann-Hilbert problem given in \cite{Peter-Sigma-2017}.

The innovation of this paper consists of the following two points: (i) We derive the rational solutions of Painlev\'{e}-II equation via two types of generalized Darboux transformation. In the first case, from the ``seed solution" $p=0, m=0$, we give the iterated Darboux transformation at the same spectrum $\lambda_1=0$ and derive the corresponding B\"{a}cklund transformation. Compared to some other rational solutions in integrable systems, the derivation of Painlev\'e-II rational solution is more difficult. To achieve it, we would need to use special limiting techniques: the Darboux transformation is iterated at the same spectrum and the spectrum $\lambda_1$ must be chosen at a special point. In the second case, we rewrite the generalized Darboux transformation by another equivalent formula and derive the rational solutions as a Gram determinant. In particular, we can analyze the asymptotics of large $y$ by this formula. (ii) Our generalized Darboux transformation can be used to construct the solution of Riemann-Hilbert problem given in \cite{Peter-Sigma-2017}, which provides a new way for solving the corresponding Riemann-Hilbert problem for the Painlev\'{e}-II equation.

The outline of this paper is as follows: In section \ref{sec:DT-1}, we give detailed derivation of rational solutions via the generalized Darboux transformation. Then, we also give the corresponding B\"{a}cklund transformation. In section \ref{sec:DT-2}, the generalized Darboux transformation is rewritten as another equivalent formula, which gives the solution in Gram determinant form, and can be used to analyze the asymptotics for large $y$. In section \ref{sec:RH}, we give a brief introduction about the Riemann-Hilbert problem given by \cite{Peter-Sigma-2017} and solve it with the aid of the Darboux transformation. The final part is the conclusion.

\section{The Darboux transformation of Painelv\'{e}-II equation}
\label{sec:DT-1}
In this section, we would like to construct the generalized Darboux transformation of Painlev\'{e}-II equation and derive the rational solutions. The rational solutions of Painlev\'{e}-II equation play an analogous 
role as the solitons of the integrable partial differential equations. However, to the best of our knowledge, the rational solutions of Painlev\'{e}-II equation have never been constructed by the well known Darboux transformation. From the Lax pair of Eq.\eqref{eq:Lax-pair}, we can see that the $\partial y$-part of Lax pair is the same as the $\partial x$-equation, what we usually call the Lax equation, in the Lax pair for the defocusing NLS equation\cite{Ablowitz-1974-study}, while the $\partial \lambda$-part is different. It has two types of singularity. One is at $\lambda=\infty$ and the other one is at $\lambda=0$. What is more, the $\partial \lambda$-equation contains a constant $m$, which leads to a new difficulty in the Darboux transformation. To overcome these difficulties, we use two kinds of limiting skills following the original idea of the so-called generalized Darboux transformation \cite{ling-PRE-2012,ling-non-2013}. We want to emphasize here that there is a difference between the rational solutions of Painlev\'{e}-II equation and the higher order rogue wave of integrable equations. The former ones are meromorphic functions with simple poles in the complex plane but the latter ones are global solutions in the $(x,t)$-plane.  The detailed derivation is shown as follows.

To construct the rational solutions, we start with the ``seed solution" $p=0, m=0$. Since the $\partial y$-part of Lax pair is the spectral problem of AKNS system, we can set the Darboux transformation matrix as
\begin{equation}\label{DT-2}
\mathbf{T}_{1}(\lambda; y)=\mathbb{I}_2-\frac{\lambda_1-\lambda_1^*}{\lambda-\lambda_1^*}\frac{\phi_{1}\phi_{1}^{\dagger}\sigma_3}{\phi_1^{\dagger}\sigma_3\phi_1},
\end{equation}
where $\phi_{1}$ is a special solution of Lax pair Eq.\eqref{eq:Lax-pair} with $\lambda=\lambda_1\in {\rm i}\mathbb{R}$, that is $\phi_{1}=\ee^{\left(-\ii\lambda_1y-2\ii\lambda_1^3\right)\sigma_3}\mathbf{c}$, and $\mathbf{c}=\begin{bmatrix}c_1, c_2
\end{bmatrix}^{{\rm T}}$ is a constant vector to be determined. We use $\sigma_3$ to denote one of the three Pauli matrices
\begin{equation}
\sigma_1=\begin{bmatrix}0&1\\1&0
\end{bmatrix}, \quad \sigma_2=\begin{bmatrix}0&-\ii\\\ii&0\end{bmatrix},\quad \sigma_3=\begin{bmatrix}1&0\\0&-1
\end{bmatrix}.
\end{equation}
For convenience, denote the entries of $\phi_{j}(j=1,2,\cdots)$ and $\phi_{j}^{[k]}(j=1,2,\cdots, k=1,2,\cdots)$ as follows:
\begin{equation}
\phi_{j}:=\begin{bmatrix}\phi_{j,1}\\
\phi_{j,2}
\end{bmatrix},\quad \phi_{j}^{[k]}:=\begin{bmatrix}\phi_{j,1}^{[k]}\\
\phi_{j,2}^{[k]}
\end{bmatrix}.
\end{equation}

The Darboux transformation is given as follows. Let $\mathbf{v}_1=\mathbf{T}_1(\lambda;y) \mathbf{v}$, then $\mathbf{v}_1$ satisfies a new Lax pair equation:
\begin{equation}\label{eq:Lax-pair-1}
\begin{split}
\frac{\partial \mathbf{v}_{1}}{\partial \lambda}&=\mathbf{A}^{\rm FN}_{1}(\lambda; y)\mathbf{v}_{1}, \quad \mathbf{A}^{\rm FN}_{1}(\lambda; y)=\mathbf{T}_{1, \lambda}(\lambda; y)\mathbf{T}^{-1}_{1}(\lambda; y)+\mathbf{T}_{1}(\lambda; y)\mathbf{A}^{\rm FN}(\lambda; y)\mathbf{T}_{1}^{-1}(\lambda; y),\\
\frac{\partial \mathbf{v}_{1}}{\partial y}&=\mathbf{B}^{\rm FN}_{1}(\lambda; y)\mathbf{v}_{1},\quad  \mathbf{B}^{\rm FN}_{1}(\lambda; y)=\mathbf{T}_{1, y}(\lambda; y)\mathbf{T}^{-1}_{1}(\lambda; y)+\mathbf{T}_{1}(\lambda; y)\mathbf{B}^{\rm FN}(\lambda; y)\mathbf{T}^{-1}_{1}(\lambda; y).
\end{split}
\end{equation}

Equipped with the Darboux transformation defined above, we know that the original Lax pair \eqref{eq:Lax-pair} is converted to a new one by replacing $p, m$ with $p^{[1]}, m_{1}$. In the transformed Lax pair we use the superscript $^{[1]}$ in $p^{[1]}$ {to} mean the first order rational solution, and the subscript in $m_1$ is the constant in correspondence with $p^{[1]}$.

Expanding $\mathbf{B}^{\rm FN}_{1}(\lambda; y)$ at $\lambda\to\infty$, the first order rational solution $p^{[1]}$ is given by
\begin{equation}\label{eq:p1p0}
\begin{split}
p^{[1]}=p+\frac{2\ii\left(\lambda_1-\lambda_1^*\right)\phi_{1,1}\phi_{1,2}^*}{|\phi_{1,1}|^2-|\phi_{1,2}|^2}=p+\frac{2\ii\left(\lambda_1-\lambda_1^*\right)\phi_{1,1}^*\phi_{1,2}}{|\phi_{1,1}|^2-|\phi_{1,2}|^2},
\end{split}
\end{equation}
which imposes one symmetric constraint for $\phi_1$: $\phi_{1,1}\phi_{1,2}^*=\phi_{1,1}^*\phi_{1,2}$.

In \cite{ling-non-2013}, the authors showed that the $\partial y$-equation in the new Lax pair is valid, i.e.\mbox{ }it has the same form except with new $p$ under this transformation. Therefore we only need to check whether the Darboux transformation also satisfies the $\partial \lambda$-part of Lax pair. Plugging the Darboux matrix Eq.\eqref{DT-2} and the original seed solution $p=0, m=0$ into the new Lax pair, we have
\begin{multline}\label{eq:eq-3}
\mathbf{A}^{\rm FN}_{1}(\lambda; y)=\frac{\lambda_1{-}\lambda_1^*}{\left(\lambda{-}\lambda_1^*\right)^2}\frac{\phi_1\phi_1^{\dagger}\sigma_3}{\phi_1^{\dagger}\sigma_3\phi_{1}}\left(\mathbb{I}_{2}{-}\frac{\lambda_1^*{-}\lambda_1}{\lambda{-}\lambda_1}\frac{\phi_1\phi_1^{\dagger}\sigma_3}{\phi_1^{\dagger}\sigma_3\phi_1}\right)
\\{-}\left(\mathbb{I}_2{-}\frac{\lambda_1{-}\lambda_1^*}{\lambda{-}\lambda_1^*}\frac{\phi_{1}\phi_{1}^{\dagger}\sigma_3}{\phi_1^{\dagger}\sigma_3\phi_1}\right)\left(6\ii\lambda^2\sigma_3{+}\ii y\sigma_3+m\lambda^{-1}\sigma_1\right)\left(\mathbb{I}_2{-}\frac{\lambda_1^*{-}\lambda_1}{\lambda{-}\lambda_1}\frac{\phi_{1}\phi_{1}^{\dagger}\sigma_3}{\phi_1^{\dagger}\sigma_3\phi_1}\right).
\end{multline}
Expanding Eq.\eqref{eq:eq-3} in the neighborhood of $\lambda=\infty$ and comparing the coefficient of the off-diagonal part of Eq.\eqref{eq:eq-3}, we know that the coefficient of $\lambda^{-1}$ is
\begin{equation}\label{eq:m1}
m_1:=\frac{4\ii\lambda_1\phi_{1,1}\phi_{1,2}^*}{|\phi_{1,1}|^2-|\phi_{1,2}|^2}y+\frac{24\ii\lambda_1^3\phi_{1,1}\phi_{1,2}^*}{|\phi_{1,1}|^2-|\phi_{1,2}|^2},
\end{equation}
which should be a constant independent of $y$. Therefore, the spectral parameter $\lambda_1$ is chosen to be zero and the entries of $\phi_1$ should satisfy $|\phi_{1,1}|=|\phi_{1,2}|$. To achieve this, we can set the constant vector $\mathbf{c}$ as $\mathbf{c}=\left(1,-1\right)^{\T}$, then $\phi_{1,1}\phi_{1,2}^*=\phi_{1,1}^*\phi_{1,2}$ and $|\phi_{1,1}|=|\phi_{1,2}|$. Note that now both the numerator and the denominator are zero. Expand for infinitesimal $\lambda_1$, we get $m_1=1$, and the first order rational solution $p^{[1]}$ is
\begin{equation}\label{eq:p1}
p^{[1]}=\frac{1}{y}.
\end{equation}
\smallskip

In order to generate the 2nd order rational solution, the second fold Darboux transformation matrix can be set as
\begin{equation}
\mathbf{T}_{2}(\lambda; y)=\mathbb{I}-\frac{\lambda_2-\lambda_2^*}{\lambda-\lambda_2^*}\frac{\phi_{2}^{[1]}\phi_{2}^{[1],\dagger}\sigma_3}{\phi_{2}^{[1],\dagger}\sigma_3\phi_{2}^{[1]}},
\end{equation}
where $\phi_{2}^{[1]}$ is defined as
\begin{equation}\label{eq:phi21}
\phi_{2}^{[1]}=\mathbf{T}_{1}(\lambda_2, y)\phi_{2},
\end{equation}
and $\phi_2=\ee^{\left(-\ii\lambda_2y-2\ii\lambda_2^3\right)\sigma_3}\mathbf{c}$. Recall $\mathbf{c}=(1,-1)^{\T}$, then $\mathbf{v}_2=\mathbf{T}_{2}(\lambda; y)\mathbf{v}_1$ satisfies the new Lax pair
\begin{equation}\label{eq:lax-2}
\begin{split}
\frac{\partial \mathbf{v}_{2}}{\partial \lambda}=\mathbf{A}^{\rm FN}_{2}(\lambda; y)\mathbf{v}_{2}, \quad \mathbf{A}^{\rm FN}_{2}(\lambda; y)=\mathbf{T}_{2, \lambda}(\lambda; y)\mathbf{T}_{2}^{-1}(\lambda; y)+\mathbf{T}_{2}(\lambda; y)\mathbf{A}_{1}^{\rm FN}(\lambda; y)\mathbf{T}_{2}^{-1}(\lambda; y),\\
\frac{\partial \mathbf{v}_{2}}{\partial y}=\mathbf{B}^{\rm FN}_{2}(\lambda; y)\mathbf{v}_{2}, \quad \mathbf{B}^{\rm FN}_{2}(\lambda; y)=\mathbf{T}_{2,y}(\lambda; y)\mathbf{T}_{2}^{-1}(\lambda; y)+\mathbf{T}_{2}(\lambda; y)\mathbf{B}^{\rm FN}_{1}(\lambda; y)\mathbf{T}_{2}^{-1}(\lambda; y),
\end{split}
\end{equation}
with new rational solution $p^{[2]}$ and new constant $m_2$. The off-diagonal term of $\mathbf{B}_2^{\mathrm{FN}}$ is exactly $p^{[2]}$, therefore the second order rational solution $p^{[2]}$ is given by
\begin{equation}
p^{[2]}=p^{[1]}+\frac{2\ii\left(\lambda_2-\lambda_2^*\right)\phi_{2,1}^{[1]}\phi_{2,2}^{[1],*}}{\left|\phi_{2,1}^{[1]}\right|^2-\left|\phi_{2,2}^{[1]}\right|^2}.
\end{equation}
To find $m_2$, expand $\mathbf{A}^{\rm FN}_{2}(\lambda; y)$ at $\lambda\to\infty$. The coefficient of $\lambda^{-1}$ is
\begin{multline}\label{eq:m-1}
m_2=\frac{4\ii\lambda_2\phi_{2,1}^{[1]}\phi_{2,2}^{[1],*}}{\left|\phi_{2,1}^{[1]}\right|^2{-}\left|\phi_{2,2}^{[1]}\right|^2}\left(3p^{[1],2}{+}y+6\lambda_2^2\right)
{-}6\ii\lambda_2\frac{\left|\phi_{2,1}^{[1]}\right|^2{+}\left|\phi_{2,2}^{[1]}\right|^2}{\left|\phi_{2,1}^{[1]}\right|^2{-}\left|\phi_{2,2}^{[1]}\right|^2}p^{[1]}_{y}\\
{+}12\lambda_2^2\frac{\left|\phi_{2,1}^{[1]}\right|^4{-}2\left(\phi^{[1]}_{2,1}\phi_{2,2}^{[1],*}\right)^2{+}\left|\phi_{2,2}^{[1]}\right|^4}{\left(\left|\phi_{2,1}^{[1]}\right|^2{-}\left|\phi_{2,2}^{[1]}\right|^2\right)^2}p^{[1]}{+}1.
\end{multline}

We first look at the first term $\frac{4\ii\lambda_2\phi_{2,1}^{[1]}\phi_{2,2}^{[1],*}}{\left|\phi_{2,1}^{[1]}\right|^2-\left|\phi_{2,2}^{[1]}\right|^2}\left(3p^{[1],2}+y+6\lambda_2^2\right)$.  Substituting Eq.\eqref{eq:p1} into Eq.\eqref{eq:m-1}, $3p^{[1],2}+y+6\lambda_2^2$ is a function of $y$. If Eq.\eqref{eq:m-1} is a constant, then the spectral parameter $\lambda_2$ must also be zero, i.e. $\lambda_2=\lambda_1=0$. When $\lambda_2=\lambda_1$, the function $\phi_{2}^{[1]}$ is degenerate. Thus the Darboux transformation is no longer valid. Following the idea in \cite{ling-PRE-2012}, we can introduce the generalized Darboux transformation and modify $\phi_{2}^{[1]}$ as $\phi_{1}^{[1]}$ with the following relation,
\begin{equation}
\phi_{1}^{[1]}=\frac{\partial \phi_{2}^{[1]}}{\partial \lambda_2}\Big|_{\lambda_2=\lambda_1}.
\end{equation}
Applying the generalized Darboux transformation, the second order rational solution $p^{[2]}$ can just be given by replacing $\phi_{2}^{[1]}$ with $\phi_{1}^{[1]}$.

\smallskip
In fact, if we want to derive the $k$th order rational solutions for the Painlev\'{e}-II equation, the spectral parameters $\lambda_k$ in the $k$th iteration will always be zero. We can construct the $k$th iteration of the Darboux transformation by continuing the above procedures. Begin by setting the $k$th step Darboux transformation matrix as
\begin{equation}
\mathbf{T}_{k}(\lambda; y)=\mathbb{I}-\frac{\lambda_1-\lambda_1^*}{\lambda-\lambda_1^*}\frac{\phi_{1}^{[k-1]}\phi_{1}^{[k-1],\dagger}\sigma_3}{\phi_{1}^{[k-1]\dagger}\sigma_3\phi_{1}^{[k-1]}},
\end{equation}
where $\phi_{1}^{[k-1]}$ is defined as
\begin{equation}\label{eq:phin1}
\phi_{1}^{[k-1]}=\frac{\partial^{k-1}\left(\mathbf{T}_{k-1}(\lambda_1, \lambda_{k}; y)\cdots\mathbf{T}_{2}(\lambda_1, \lambda_{k}; y)\mathbf{T}_{1}(\lambda_1, \lambda_{k}; y)\phi_{k}\right)}{\partial (\lambda_{k})^{k-1}}\Big|_{\lambda_{k}=\lambda_1},
\end{equation}
with $\phi_k=\ee^{\left(-\ii\lambda_ky-2\ii\lambda_{k}^3\right)\sigma_3}\mathbf{c}$. Then the $k$th order rational solution $p^{[k]}$ is given by
\begin{equation}\label{eq:bt-2}
p^{[k]}=p^{[k-1]}+\frac{2\ii\left(\lambda_1-\lambda_1^*\right)\phi_{1,1}^{[k-1]}\phi_{1,2}^{[k-1],*}}{\left|\phi_{1,1}^{[k-1]}\right|^2-\left|\phi_{1,2}^{[k-1]}\right|^2}.
\end{equation}
Correspondingly, the $k$th transformed Lax pair becomes
\begin{equation}
\begin{split}
\frac{\partial \mathbf{v}_{k}}{\partial \lambda}=\mathbf{A}^{\rm FN}_{k}(\lambda; y)\mathbf{v}_{k}, \quad \mathbf{A}^{\rm FN}_{k}(\lambda; y)=\mathbf{T}_{k, \lambda}(\lambda; y)\mathbf{T}_{k}^{-1}(\lambda; y)+\mathbf{T}_{k}(\lambda; y)\mathbf{A}_{k-1}^{\rm FN}(\lambda; y)\mathbf{T}_{k}^{-1}(\lambda; y),\\
\frac{\partial \mathbf{v}_{k}}{\partial y}=\mathbf{B}^{\rm FN}_{k}(\lambda; y)\mathbf{v}_{k}, \quad \mathbf{B}^{\rm FN}_{k}(\lambda; y)=\mathbf{T}_{k,y}(\lambda; y)\mathbf{T}_{k}^{-1}(\lambda; y)+\mathbf{T}_{k}(\lambda; y)\mathbf{B}^{\rm FN}_{k-1}(\lambda; y)\mathbf{T}_{k}^{-1}(\lambda; y).
\end{split}
\end{equation}
Similar as above, the coefficient of $\lambda^{-1}$ in the off-diagonal elements in $\mathbf{A}_{k}^{\rm FN}(\lambda; y)$ is
\begin{multline}\label{eq:eq-6}
m_k=2\ii\left(\lambda_1-\lambda_1^*\right)\frac{\phi_{1,1}^{[k-1]}\phi_{1,2}^{[k-1],*}}{\left|\phi_{1,1}^{[k-1]}\right|^2-\left|\phi_{1,2}^{[k-1]}\right|^2}\left[3\left(p^{[k-1]}\right)^2+y\right]
-3\ii\left(\lambda_1-\lambda_1^*\right) \frac{\left|\phi_{1,1}^{[k-1]}\right|^2+\left|\phi_{1,2}^{[k-1]}\right|^2}{\left|\phi_{1,1}^{[k-1]}\right|^2-\left|\phi_{1,2}^{[k-1]}\right|^2}p^{[k-1]}_{y}+m_{k-1}.
\end{multline}
Substituting Eq.\eqref{eq:bt-2} into Eq.\eqref{eq:eq-6}, we have
\begin{equation}\label{eq:eq-7}
m_k=-3\ii\left(\lambda_1-\lambda_1^*\right) \frac{\left|\phi_{1,1}^{[k-1]}\right|^2+\left|\phi_{1,2}^{[k-1]}\right|^2}{\left|\phi_{1,1}^{[k-1]}\right|^2-\left|\phi_{1,2}^{[k-1]}\right|^2}p^{[k-1]}_{y}+
\left(p^{[k]}-p^{[k-1]}\right)\left[3\left(p^{[k-1]}\right)^2+y\right]+m_{k-1}.
\end{equation}
A transformation from $p^{[k-1]}$ to $p^{[k]}$ is usually called a B\"acklund tranfsformation, therefore Eq.\eqref{eq:eq-7} can be regarded as an incomplete B\"{a}cklund transformation, with the coefficient of $p_y^{[k-1]}$ to be determined. If in Eq.\eqref{eq:eq-7}, the coefficient of $p_y^{[k-1]}$  can also be written as a relation between $p^{[k]}$ and $p^{[k-1]}$, then Eq.\eqref{eq:eq-7} will complete our transformation. Next, we give a lemma about the symmetry of $\phi_{1}^{[k-1]}$, which will help us complete the B\"{a}cklund transformation.
\begin{lemma}\label{prop:sym}
Let $\phi_k=e^{\left(-\ii\lambda_ky-2\ii\lambda_k^3 \right)\sigma_3}\mathbf{c}$, and the constant vector $\mathbf{c}=\left(1,-1\right)^{\T}$. Define the vector $\phi_{k}^{[k-1]}$ as
\begin{equation}\label{eq:phikk-1}
\phi_{k}^{[k-1]}=\mathbf{T}_{k-1}(\lambda_1, \lambda_{k}; y)\cdots\mathbf{T}_{1}(\lambda_1, \lambda_{k}; y)\mathbf{T}_{0}(\lambda_1, \lambda_{k}; y)\phi_{k},\,\, (k=1,2,\cdots), \quad \mathbf{T}_{0}(\lambda_1, \lambda_{k}; y)=\mathbb{I}_2.
\end{equation}
Then for $\lambda_i$, $\lambda_k\in\mathrm{i}\mathbb{R}$, $\phi_k^{[k-1]}$ satisfies the following symmetry relation
\begin{equation}\label{eq:symme-phim-1}
\begin{split}
&\phi_{k,1}^{[k-1]}(\lambda_1,\lambda_k; y)=\left[\phi_{k,1}^{[k-1]}(\lambda_1, \lambda_k; y)\right]^*=-\left[\phi_{k,2}^{[k-1]}(\lambda_1^*, \lambda_k^*; y)\right]^*, \\
&\phi_{k,2}^{[k-1]}(\lambda_1, \lambda_k; y)=\left[\phi_{k,2}^{[k-1]}(\lambda_1, \lambda_k); y\right]^*.
\end{split}
\end{equation}
In comparison, the vector $\phi_{1}^{[k-1]}$ in generalized Darboux transformation defined in Eq.\eqref{eq:phin1} satisfies a different kind of symmetry relation,
\begin{equation}\label{eq:symme-phim-2}
\phi^{[k-1]}_{1,1}(\lambda_1;y)=(-1)^{k-1}\phi^{[k-1],*}_{1,1}(\lambda_1;y)=(-1)^{k}\phi^{[k-1]}_{1,2}(\lambda_1^*;y)=-\phi^{[k-1],*}_{1,2}(\lambda_1^*;y),
\end{equation}
with $\lambda_1\in\mathrm{i}\mathbb{R}$.
\end{lemma}
\begin{proof}
Using induction. When $k=1, \lambda_1=0$, $\phi_{1}=\left(\ee^{-\ii\lambda_1y-2\ii\lambda_1^3},-\ee^{\ii\lambda_1y+2\ii\lambda_1^3}\right)^{\T}$ immediately satisfies the symmetry relation Eq.\eqref{eq:symme-phim-1} and Eq.\eqref{eq:symme-phim-2}. Supposing that these symmetries are valid under the $k-1$th step Darboux transformation, we only need to check that in the $k$th step, these symmetries are still valid. To prove it, we rewrite Eq.\eqref{eq:phikk-1} as
\begin{equation}
\phi_{k}^{[k-1]}=\mathbf{T}_{k-1}(\lambda_1, \lambda_k; y)\widehat{\phi}_{k}^{[k-2]},\quad \widehat{\phi}_{k}^{[k-2]}:=\mathbf{T}_{k-2}(\lambda_1, \lambda_{k}; y)\cdots\mathbf{T}_{1}(\lambda_1, \lambda_{k}; y)\phi_{k},
\end{equation}
where
$\mathbf{T}_{k-1}(\lambda_1, \lambda_k;y)=\mathbb{I}_2-\frac{\lambda_{1}-\lambda_{1}^*}{\lambda_k-\lambda_{1}^*}\frac{\phi_{1}^{[k-2]}\phi_{1}^{[k-2],\dagger}\sigma_3}{\phi_{1}^{[k-2],\dagger}\sigma_3\phi_{1}^{[k-2]}}$, $\phi_{1}^{[k-2]}=\frac{\partial^{k-2}\left(\mathbf{T}_{k-2}(\lambda_1, \lambda_{k-1};y)\cdots\mathbf{T}_{1}(\lambda_1, \lambda_{k-1};y)\phi_{k-1}\right)}{\partial (\lambda_{k-1})^{k-2}}\Big|_{\lambda_{k-1}=\lambda_1}.$ By a simple calculation, $\phi_{k}^{[k-1]}$ is equal to
\begin{equation}
\phi_{k}^{[k-1]}=\widehat{\phi}_{k}^{[k-2]}-\left(\frac{\lambda_{1}-\lambda_{1}^*}{\lambda_k-\lambda_{1}^*}\frac{\phi_{1}^{[k-2],\dagger}\sigma_3\widehat{\phi}_{k}^{[k-2]}}{\phi_{1}^{[k-2],\dagger}\sigma_3\phi_{1}^{[k-2]}}\right)\phi_{1}^{[k-2]}.
\end{equation}
It is clear that if $\widehat{\phi}_{k}^{[k-2]}$ and $\phi_{1}^{[k-2]}$ satisfy the symmetry relation Eq.\eqref{eq:symme-phim-1} and Eq.\eqref{eq:symme-phim-2} respectively under this assumption, then $\phi_{k}^{[k-1]}$ will also satisfy the symmetry Eq.\eqref{eq:symme-phim-1} when $\lambda_{k}\in\ii\mathbb{R}$.

Next, we give the proof of Eq.\eqref{eq:symme-phim-2}. From the definition of $\phi_{1}^{[k-1]}$ in Eq.\eqref{eq:phin1}, we know that the entries of $\phi_{1}^{[k-1]}$ can be written as
\begin{equation}\label{eq:symme-phim}
\begin{split}
\phi_{1,1}^{[k-1]}(\lambda_1;y)=\frac{\partial^{k-1}\phi_{k,1}^{[k-1]}(\lambda_1, \lambda_k; y)}{\partial (\lambda_{k})^{k-1}}\Big|_{\lambda_{k}=\lambda_1}, \quad \phi_{1,1}^{[k-1],*}(\lambda_1;y)=\frac{\partial^{k-1}\phi_{k,1}^{[k-1],*}(\lambda_1, \lambda_k; y)}{\partial (\lambda_{k}^*)^{k-1}}\Big|_{\lambda_{k}^*=\lambda_1^*},\\
\phi_{1,2}^{[k-1]}(\lambda_1^*;y)=\frac{\partial^{k-1}\phi_{k,2}^{[k-1]}(\lambda_1^*, \lambda_k^*; y)}{\partial (\lambda_{k})^{k-1}}\Big|_{\lambda_{k}=\lambda_1},\quad \phi_{1,2}^{[k-1],*}(\lambda_1^*;y)=\frac{\partial^{k-1}\phi_{k,2}^{[k-1],*}(\lambda_1^*,\lambda_k^*;y)}{\partial (\lambda_{k}^*)^{k-1}}\Big|_{\lambda_{k}^*=\lambda_1^*},
\end{split}
\end{equation}
which indicates $\phi_{1,1}^{[k-1]}(\lambda_1;y)=(-1)^{k-1}\phi_{1,1}^{[k-1],*}(\lambda_1;y)=(-1)^{k}\phi_{1,2}^{[k-1]}(\lambda_1^*;y)=-\phi_{1,2}^{[k-1],*}(\lambda_1^*;y)$, it completes the proof.
\end{proof}
Based on the symmetry relation in lemma \ref{prop:sym}, when $\lambda_1=0$, we have
\begin{equation}\label{eq:symme-phim-3}
\phi_{1,1}^{[k-1]}(0;y)=(-1)^{k-1}\phi_{1,1}^{[k-1],*}(0;y)=(-1)^{k}\phi_{1,2}^{[k-1]}(0;y)=-\phi_{1,2}^{[k-1],*}(0;y),
\end{equation}
substituting Eq.\eqref{eq:symme-phim-3} into Eq.\eqref{eq:eq-7}, we can get an explicit relation between $p^{[k]}$ and $p^{[k-1]}$,
\begin{equation}\label{eq:bt-1}
\left(p^{[k]}-p^{[k-1]}\right)\left((-1)^{k-1}3p^{[k-1]}_{y}+3\left(p^{[k-1]}\right)^2+y\right)+m_{k-1}=m_{k}.
\end{equation}
Then the B\"{a}cklund transformation between $p^{[k]}$ and $p^{[k-1]}$ is
\begin{equation}\label{eq:BT-P-II-1}
p^{[k]}=p^{[k-1]}+\frac{m_{k}-m_{k-1}}{(-1)^{k-1}3p_{y}^{[k-1]}+3\left(p^{[k-1]}\right)^2+y}.
\end{equation}

Summarizing the above calculation, we obtain rational solutions of all orders. In \cite{Airault-Study-1979}, Airault pointed out that $m\in\mathbb{Z}$ is the necessary and sufficient condition for Painlev\'e-II solution to be rational. By uniqueness, we conclude all general rational solutions are given by our generalized Darboux transform in Theorem \ref{theo-1}.
\begin{theorem}\label{theo-1}
Setting $\lambda_1=-\lambda_1^*=0$, the generalized Darboux transformation matrices for the Painlev\'{e}-II equation \eqref{eq:P-II} are given by
\begin{equation}\label{eq:DT}
\mathbf{T}^{[k]}(\lambda; y)=\mathbf{T}_{k}(\lambda; y)\cdots\mathbf{T}_{2}(\lambda; y)\mathbf{T}_{1}(\lambda; y),
\end{equation}
where
\begin{equation}
\mathbf{T}_{k}(\lambda; y)=\mathbb{I}-\frac{\lambda_1-\lambda_1^*}{\lambda-\lambda_1^*}\frac{\phi_{1}^{[k-1]}\phi_{1}^{[k-1],\dagger}\sigma_3}{\phi_{1}^{[k-1],\dagger}\sigma_3\phi_{1}^{[k-1]}},
\end{equation}
and $\phi_{1}^{[k-1]}$ is defined in Eq.\eqref{eq:phin1}. By expanding $\mathbf{T}^{[k]}(\lambda; y)$ in neighborhood of $\lambda=\infty$, the rational solutions of Painlev\'{e}-II Eq.\eqref{eq:P-II} can be given as
\begin{equation}\label{eq:pm}
p^{[k]}(y)=2\ii\lim\limits_{\lambda\to\infty}\lambda \mathbf{T}^{[k]}(\lambda; y)_{12}=p+\sum_{l=0}^{k-1}\frac{2\ii\left(\lambda_1-\lambda_1^*\right)\phi_{1,1}^{[l]}\phi_{1,2}^{[l],*}}{\left|\phi_{1,1}^{[l]}\right|^2-\left|\phi_{1,2}^{[l]}\right|^2},
\end{equation}
where $p$ is the seed solution $p=0$ in \eqref{DT-2}.
\end{theorem}

\bigskip
\noindent We remark on the specificity of Painlev\'e-II rational solutions in the transformation.
\begin{remark}
It is clear that each Painlev\'e-II equation Eq.\eqref{eq:P-II} involves one constant $m$, and that from our transformation in each iteration, we take a different value of $m_k$. Therefore, the B\"{a}cklund transformation between two rational solutions is not the auto-B\"{a}cklund transformation. For each order of the rational solutions, the constants $m_k$ differ. This makes our B\"acklund transformation different than what usually appears in integrable PDEs.
\end{remark}

We derived the B\"acklund transformation \eqref{eq:BT-P-II-1} through generalized Darboux transformation. In fact, apart from this method, we can also give a B\"{a}cklund transformation for $p^{[k]}$ from the Eq.\eqref{eq:P-II} itself. Before discussing it, we first give a proposition about the Schwarzian derivative.
\begin{prop}\label{remark-SD}
The Schwarzian derivative $S(\sigma)=\frac{d}{dy}\left(\frac{\sigma_{yy}}{\sigma_{y}}\right)-\frac{1}{2}\left(\frac{\sigma_{yy}}{\sigma_{y}}\right)^2$\cite{Weiss-JMP-1983} is invariant under the M\"{o}bius group $\sigma=\frac{a\eta+b}{c\eta+d}$, where $ad-bc\neq 0$.
\end{prop}
\begin{proof}
From the M\"{o}bius transformation $\sigma=\frac{a\eta+b}{c\eta+d}$, we have
\begin{equation}\label{eq:Mobuis-T}
c\sigma\eta+d\sigma-a\eta-b=0.
\end{equation}
Taking the derivative of Eq.\eqref{eq:Mobuis-T} with respect to $y$ for once, twice and thrice respectively, then we have
\begin{equation}
\begin{bmatrix}
(\sigma\eta)_{y}&\sigma_{y}&-\eta_{y}\\
(\sigma\eta)_{yy}&\sigma_{yy}&-\eta_{yy}\\
(\sigma\eta)_{yyy}&\sigma_{yyy}&-\eta_{yyy}
\end{bmatrix}
\begin{bmatrix}c\\d\\a
\end{bmatrix}=0,
\end{equation}
which indicates
\begin{equation}
\left|\begin{matrix}(\sigma\eta)_{y}&\sigma_{y}&-\eta_{y}\\
(\sigma\eta)_{yy}&\sigma_{yy}&-\eta_{yy}\\
(\sigma\eta)_{yyy}&\sigma_{yyy}&-\eta_{yyy}
\end{matrix}\right|=0.
\end{equation}
Through a simple calculation, we can get the relation between $\sigma$ and $\eta$,
\begin{equation}
\begin{split}
&2\sigma_{y}\eta_{y}^2\sigma_{yyy}-3\eta_{y}^2\sigma_{yy}^2-2\eta_{y}\sigma_{y}^2\eta_{yyy}+3\sigma_{y}^2\eta_{yy}^2=0,\\
\Rightarrow&\frac{2\sigma_{yyy}}{\sigma_{y}}-\frac{3}{2}\left(\frac{\sigma_{yy}}{\sigma_{y}}\right)^2=\frac{2\eta_{yyy}}{\eta_{y}}-\frac{3}{2}\left(\frac{\eta_{yy}}{\eta_{y}}\right)^2,
\end{split}
\end{equation}
which is exactly the Schwarzian derivative. This completes the proof.
\end{proof}

Next, we give the B\"{a}cklund transformation in another way following the idea in \cite{Clarkson-IP-1999}.
For convenience, we rewrite Eq.\eqref{eq:P-II} in equivalent form,
\begin{equation}\label{eq:P-II-1}
\left(\frac{d}{dy}+2p\right)\left(\frac{dp}{dy}-p^2\right)-\frac{2}{3}py+\frac{2}{3}m=0.
\end{equation}
Setting $p=\frac{1}{2}\frac{\sigma_{yy}}{\sigma_{y}}$ and substituting into Eq.\eqref{eq:P-II-1}, we have
\begin{equation}\label{eq:P-II-2}
\left(\frac{d}{dy}+\frac{\sigma_{yy}}{\sigma_{y}}\right)\left[\frac{1}{2}\frac{\sigma_{yyy}}{\sigma_{y}}-\frac{3}{4}\left(\frac{\sigma_{yy}}{\sigma_{y}}\right)^2\right]-\frac{1}{3}\frac{\sigma_{yy}}{\sigma_{y}}y+\frac{2}{3}m=0.
\end{equation}
Notice $\frac{1}{2}\frac{\sigma_{yyy}}{\sigma_{y}}-\frac{3}{4}\left(\frac{\sigma_{yy}}{\sigma_{y}}\right)^2$ is the Schwarzian derivative of $\sigma$, which is invariant under the M\"{o}bius transformation (Proposition \ref{remark-SD}). Thus, it is natural to introduce the M\"{o}bius transformation group $\sigma=\frac{a\eta+b}{c\eta+d}$, $ad-bc\neq 0$, and substitute it into Eq.\eqref{eq:P-II-2},
\begin{equation}\label{eq:P-II-3}
\left(\frac{d}{dy}+\frac{\eta_{yy}}{\eta_{y}}-\frac{2c\eta_{y}}{c\eta+d}\right)\left[\frac{1}{2}\frac{\eta_{yyy}}{\eta_{y}}-\frac{3}{4}\left(\frac{\eta_{yy}}{\eta_{y}}\right)^2\right]-\frac{1}{3}\left(\frac{\eta_{yy}}{\eta_{y}}-\frac{2c\eta_{y}}{c\eta+d}\right)y+\frac{2}{3}m=0.
\end{equation}
Moreover, imposing that Eq.\eqref{eq:P-II-3} satisfiies the transformed Painlev\'e-II equation, i.e. it has the same form as Eq.\eqref{eq:P-II-2}, except with a different constant $\tilde{m}$,
\begin{equation}\label{eq:P-II-4}
\left(\frac{d}{dy}+\frac{\eta_{yy}}{\eta_{y}}\right)\left[\frac{1}{2}\frac{\eta_{yyy}}{\eta_{y}}-\frac{3}{4}\left(\frac{\eta_{yy}}{\eta_{y}}\right)^2\right]-\frac{1}{3}\left(\frac{\eta_{yy}}{\eta_{y}}\right)y+\frac{2}{3}\tilde{m}=0,
\end{equation}
then the identity between the difference of two Painlev\'e-II solutions, $\frac{1}{2}\frac{\sigma_{yy}}{\sigma_{y}}-\frac{1}{2}\frac{\eta_{yy}}{\eta_{y}}=-\frac{c\eta_{y}}{c\eta+d}$, can be regarded as the B\"{a}cklund transformation for Painlev\'e-II solutions with respect to different constants $m$ and $\widetilde{m}$.
Subtracting Eq.\eqref{eq:P-II-4} from Eq.\eqref{eq:P-II-3}, we have
\begin{equation}\label{eq:BT-P-II}
\left(\frac{2c\eta_{y}}{c\eta+d}\right)\left[\frac{1}{2}\frac{\eta_{yyy}}{\eta_{y}}-\frac{3}{4}\left(\frac{\eta_{yy}}{\eta_{y}}\right)^2\right]-\frac{2}{3}\frac{c\eta_{y}}{c\eta+d}y+\frac{2}{3}\left(\tilde{m}-m\right)=0.
\end{equation}
Equivalently,
\begin{equation}\label{eq:P-II-5}
\frac{1}{2}\frac{\eta_{yyy}}{\eta_{y}}-\frac{3}{4}\left(\frac{\eta_{yy}}{\eta_{y}}\right)^2=\frac{(m-\tilde{m})(c\eta+d)}{3c\eta_{y}}+\frac{1}{3}y.
\end{equation}
Next, substitute Eq.\eqref{eq:P-II-5} into Eq.\eqref{eq:P-II-4},
\begin{equation}
\left(\frac{d}{dy}+\frac{\eta_{yy}}{\eta_{y}}\right)\left(\frac{(m-\tilde{m})(c\eta+d)}{3c\eta_{y}}+\frac{1}{3}y\right)-\frac{1}{3}\frac{\eta_{yy}}{\eta_{y}}y+\frac{2}{3}\tilde{m}=0.
\end{equation}
Direct computation shows the relation between $m$ and $\tilde{m}$, that is $m=-1-\tilde{m}$. Plugging $\tilde{p}=\frac{1}{2}\frac{\eta_{yy}}{\eta_{y}}$ into Eq.\eqref{eq:BT-P-II}, the B\"{a}cklund transformation factor $\frac{-c\eta_{y}}{c\eta+d}$ is thus given by
\begin{equation}\label{eq:BT-P-II-2}
\frac{-c\eta_{y}}{c\eta+d}=p-\tilde{p}=-\frac{\left(1+2\tilde{m}\right)}{-3\tilde{p}_{y}+3\tilde{p}^2+y}.
\end{equation}
Furthermore, from this B\"{a}cklund transformation \eqref{eq:BT-P-II-2} and the discrete symmetry $\left(p(y),m\right)\rightarrow \left(-p(y),-m\right)$, the B\"{a}cklund transformation about the Painlev\'{e}-II equation can be given by
\begin{equation}
p=\tilde{p}-\frac{2\tilde{m}\pm1}{\mp3\tilde{p}_y+3\tilde{p}^2+y},\quad m=\mp 1-\tilde{m}.
\end{equation}

In the following section, we will prove in Theorem  \ref{theo-2} that the constant $m_k$ corresponding to the $k$th order rational solution $p^{[k]}$ is $m_{k}=(-1)^{k-1}k$. If we set $(-1)^{k-2}p^{[k-1]}=\tilde{p}, (-1)^{k-2}p^{[k]}=p$, then the two B\"{a}cklund transformations Eq.\eqref{eq:BT-P-II-1} and Eq.\eqref{eq:BT-P-II-2} are equivalent.

\section{The asymptotics for large $y$}
\label{sec:DT-2}
In the last section, through generalized Darboux transformations \eqref{eq:DT}, we obtain the rational solutions in Eq.\eqref{eq:pm}. The B\"{a}cklund transformation \eqref{eq:BT-P-II-1} is also directly given from iteration. In this section, we will continue to discuss some properties of the rational solutions of Painlev\'{e}-II equation. So far, the rational solutions in Eq.\eqref{eq:pm} are in abstract form, which does not lend itself to easy analysis of their properties. In this section, following the idea in \cite{Ling-Phyd-2016}, we convert the generalized Darboux transformation \eqref{eq:DT} into another equivalent formulation and rewrite the rational solutions as a Gram determinant. A great benefit of this type of formula is that it can be easily used to analyze the asymptotics for large $y$. The detailed calculation is shown as follows in Corollary \ref{theo-DT}.
\begin{corollary}\label{theo-DT}
The generalized Darboux transformation matrices \eqref{eq:DT} can be rewritten as the following equivalent formula
\begin{equation}\label{eq:darboux}
\mathbf{T}^{[k]}(\lambda;y)=\mathbb{I}-\mathbf{Y}_k\mathbf{M}^{-1}\mathbf{D}\mathbf{Y}_k^{\dagger}\sigma_3,\qquad 
\end{equation}
where
\begin{equation}\label{DT-ele}
\begin{split}
\mathbf{Y}_k=&\left[\phi^{[0]},\phi^{[1]},\cdots,\phi^{[k-1]}\right], \quad
\mathbf{D}=\begin{bmatrix}
 \frac{1}{\lambda}&0 &\cdots & 0 \\
\frac{1}{\lambda^2}&\frac{1}{\lambda} &\cdots & 0 \\
\vdots &\vdots & &\vdots \\
\frac{1}{\lambda^{k}}&\frac{1}{\lambda^{k-1}} &\cdots &\frac{1}{\lambda} \\
\end{bmatrix},
\end{split}
\end{equation}
with $\phi^{[l]}=\frac{1}{l!}\frac{\partial^{l}\phi_{1}}{\partial \lambda_1^l}\Big|_{\lambda_1=0}$, and the elements $M_{ij}$ in $\mathbf{M}$ are defined by expanding $\frac{\phi_{1}^{\dagger}\sigma_3\phi_{1}}{\lambda_1-\lambda_1^*}$ at $\lambda_1=-\lambda_1^*=0$, that is
\begin{equation}\label{eq:M}
\begin{split}
\frac{\phi_{1}^{\dagger}\sigma_3\phi_{1}}{\lambda_1-\lambda_1^*}&=\frac{-2\ii\sin\left[\left(\lambda_1-\lambda_1^*\right)\left(2\lambda_1^2+2\lambda_1^{*,2}+2\lambda_1\lambda_1^*+y\right)\right]}{\lambda_1-\lambda_1^*}
\\&=\sum_{l=1}^{[\frac{k+1}{2}]}\frac{(-2\ii)(-1)^{l-1}(\lambda_1-\lambda_1^*)^{2l-2}\left(2\lambda_1^2+2\lambda_1^{*,2}+2\lambda_1\lambda_1^*+y\right)^{2l-1}}{(2l-1)!}+\mathscr{O}(\lambda_1-\lambda_1^*)^{2\left[\frac{k+1}{2}\right]}\\
&:=\sum_{i=1}^{k}\sum_{j=1}^{k}M_{ij}\lambda_1^{i-1}\lambda_{1}^{*,j-1}+\mathscr{O}(\lambda_1^{k+1},\lambda_1^{*,k+1}).
\end{split}
\end{equation}
Then the $k$th order rational solutions of Painlev\'{e}-II equation can be given by
\begin{equation}\label{eq:pn}
p^{[k]}=2\ii\frac{\det(\mathbf{G})}{\det(\mathbf{M})},
\end{equation}
where $\mathbf{G}=\begin{bmatrix}\mathbf{M}&Y_{k, 2}^{+}\\-Y_{k, 1}&0
\end{bmatrix}$, and $Y_{k, j}(j=1,2)$ represents the $j$th row of $\mathbf{Y}_{k}$.
\end{corollary}
\begin{proof}
A simple way to obtain the generalized Darboux transformation is by using the limiting technique.
We use an alternative form of the $k$-fold Darboux transformation matrix, see \cite{Ling-Phyd-2016}
\begin{equation}\label{DT-3}
\widehat{\mathbf{T}}^{[k]}(\lambda; y)=\mathbb{I}+\widehat{\mathbf{Y}}\widehat{\mathbf{M}}^{-1}\widehat{\mathbf{D}}^{-1}\widehat{\mathbf{Y}}^{\dagger}\sigma_3,
\end{equation}
where $\widehat{\mathbf{Y}}=\left[\phi_{1}, \phi_{2}, \cdots, \phi_{k}\right]$, $\phi_{k}$ is the special solution of Lax pair Eq.\eqref{eq:Lax-pair}, $\phi_k=\mathrm{e}^{(-\mathrm{i}\lambda_ky-2\mathrm{i}\lambda_k^3)\sigma_3}(1,-1)^{\mathsf{T}}$, with $\lambda=\lambda_k\in {\rm i}\mathbb{R}$, and $\widehat{\mathbf{M}}=\left(\frac{\phi_i^{\dagger}\sigma_3\phi_{j}}{\lambda_i^*-\lambda_j}\right)_{1\leq,i,j\leq k}, \widehat{\mathbf{D}}={\rm diag}\left(\lambda-\lambda_1^*, \lambda-\lambda_2^*, \cdots, \lambda-\lambda_k^*\right)$.

Recall the rational solutions of Painlev\'{e}-II equation is recovered by Eq.\eqref{eq:pm}.
We have proven that all the spectral parameters $\lambda_i$ $(i=1,2,\cdots, k)$ are equal to $0$, so we set the spectral parameters $\lambda_2, \lambda_3, \cdots, \lambda_{k}$ as
\begin{equation}
\lambda_2=\epsilon_{2},\quad \lambda_3=\epsilon_{3}, \cdots, \lambda_{k}=\epsilon_{k},\quad \phi_{2}=\phi_1(\epsilon_2), \quad \phi_{3}=\phi_{1}(\epsilon_3),\quad \phi_{k}=\phi_{1}(\epsilon_k).
\end{equation}
Directly substitute the above to the Darboux matrix \eqref{DT-3} and take the limit $\epsilon_{i}\to 0$, $i=2,3, \cdots, k$, then we can get the formula \eqref{eq:darboux}. Correspondingly, the rational solution of Painlev\'{e}-II equation can be given with the following formula:

\begin{equation}
p^{[k]}=2\ii\lim\limits_{\lambda\to\infty}\lambda\mathbf{T}^{[k]}_{12}(\lambda; y),
\end{equation}
which equals to Eq.\eqref{eq:pn}, it completes the proof.
\end{proof}

Apparently the rational solutions in Eq.\eqref{eq:pm} and Eq.\eqref{eq:pn} have different forms, each of which has its own advantage. The former can be used to derive the B\"{a}cklund transformation, to which the latter does not have an obvious connection. However, the latter writes all rational Painlev\'e-II solution in a compact formula. Moreover, the formula is made up of determinants, which renders analytical properties for large $y$ asymptotics very accessible. Using the new representation, we find the constant $m_k$ in Eq.\eqref{eq:P-II} for the $k$th order rational solution $p^{[k]}$, thus also completing the B\"acklund transformation Eq.\eqref{eq:BT-P-II-1}.
\begin{theorem}\label{theo-2}
The constant $m_k$ in the Painlev\'e-II equation Eq.\eqref{eq:P-II}, corresponding to the $k$th order rational solutions $p^{[k]}$ in Eq.\eqref{eq:pn} is equal to $(-1)^{k-1}k$.
\end{theorem}
\begin{proof}
Following the idea in \cite{Zhang-Phyd-2021}, we first expand $p^{[k]}$ Eq.\eqref{eq:pn} at the neighborhood of $y=\infty$ and substitute this series into the Painlev\'{e}-II equation \eqref{p-II}. Then the constant $m_k$ can be calculated by comparing the coefficients of $y$ polynomials. Clearly, the rational solution Eq.\eqref{eq:pn}  is the quotient of two Gram determinants, thus we can look at the denominator and the numerator separately. Firstly, we discuss the asymptotics of the denominator $\det(\mathbf{M})$. From the definition of $M_{ij}$ in Eq.\eqref{eq:M}, we know that while $M_{ij}$ does not generally admit simple formula, the leading behavior for large $y$ is quite easy to compute:
\begin{equation}\label{eq:aym-M}
M_{ij}=\left\{\begin{aligned}&(-2\ii y)(-1)^{\frac{i-j}{2}}\frac{y^{i+j-2}}{(i-1)!(j-1)!(i+j-1)}+\mathscr{O}(y^{i+j-5}), & \quad i+j\,\, \text{is\,\, even}, \\
&0,& \quad i+j\,\, \text{is\,\, odd}.\;\end{aligned}\right.
\end{equation}
Based on this asymptotical expression in Eq.\eqref{eq:aym-M}, when $y$ is large and $k$ is even, the denominator $\det\left(\mathbf{M}\right)$ can be expanded in the following formula:

\begin{align}
\begin{split}
\det(\mathbf{M})
&=\left|\begin{matrix}
-2\ii y&0&\cdots&\frac{-2\ii y(-\ii y)^{k-2}}{(k-2)!(k-1)}&0\\
0&\frac{2\ii y(-\ii y)^2}{3}&\cdots&0&\frac{2\ii y(-\ii y)^{k}}{(k-1)!(k+1)}\\
\frac{-2\ii y(-\ii y)^2}{2!3}&0&\cdots&\frac{-2\ii y(-\ii y)^{k}}{2!(k-2)!(k+1)}&0\\
0&\frac{2\ii y(-\ii y)^{4}}{3!5}&\cdots&0&\frac{2\ii y(-\ii y)^{k+2}}{3!(k-1)!(k+3)}\\
\vdots&\vdots&\vdots&\vdots&\vdots\\
\frac{-2\ii y(-\ii y)^{k-2}}{(k-2)!(k-1)}&0&\cdots&\frac{-2\ii y(-\ii y)^{2k-4}}{(k-2)!(k-2)!(2k-3)}&0\\
0&\frac{2\ii y(-\ii y)^{k}}{(k-1)!(k+1)}&\cdots&0&\frac{2\ii y(-\ii y)^{2k-2}}{(k-1)!(k-1)!(2k-1)}
\end{matrix}\right|+\mathcal{O}(y^{k^2-1}) \\
&=\left|\begin{matrix}-2\ii y&\cdots&\frac{-2\ii y(-\ii y)^{k-2}}{(k-2)!\left(k-1\right)}&0&\cdots&0\\
\frac{-2\ii y(-\ii y)^2}{2!3}&\cdots&\frac{-2\ii y(-\ii y)^{k}}{2!(k-2)!(k+1)}&0&\cdots&0\\
\vdots&\vdots&\vdots&\vdots&\vdots&\vdots\\
\frac{-2\ii y(-\ii y)^{k-2}}{(k-2)!(k-1)}&\cdots&\frac{-2\ii y(-\ii y)^{2k-4}}{((k-2)!)^2(2k-3)}&0&\cdots&0\\
0&\cdots&0&\frac{2\ii y(-\ii y)^2}{3}&\cdots&\frac{2\ii y(-\ii y)^{k}}{(k-1)!(k+1)}\\
0&\cdots&0&\frac{2\ii y(-\ii y)^{4}}{3!5}&\cdots&\frac{2\ii y(-\ii y)^{k+2}}{3!(k-1)!(k+3)}\\
0&\cdots&0&\frac{2\ii y(-\ii y)^{k}}{(k-1)!(k+1)}&\cdots&\frac{2\ii y(-\ii y)^{2k-2}}{((k-1)!)^2(2k-1)}
\end{matrix}\right|+\mathcal{O}(y^{k^2-1})\\
&=\frac{2^{k}(-\ii y)^{k^2}}{\left(2!3!\cdots(k-1)!\right)^2}\left|\begin{matrix}1&\frac{1}{3}&\cdots&\frac{1}{k-1}&0&\cdots&\cdots&0\\
\frac{1}{3}&\frac{1}{5}&\cdots&\frac{1}{k+1}&0&\cdots&\cdots&0\\
\vdots&\vdots&\vdots&\vdots&\vdots&\vdots&\vdots&\vdots\\
\frac{1}{k-1}&\frac{1}{k+1}&\cdots&\frac{1}{2k-3}&0&\cdots&\cdots&0\\
0&\cdots&\cdots&0&-\frac{1}{3}&-\frac{1}{5}&\cdots&-\frac{1}{k+1}\\
0&\cdots&\cdots&0&-\frac{1}{5}&-\frac{1}{7}&\cdots&-\frac{1}{k+3}\\
\vdots&\vdots&\vdots&\vdots&\vdots&\vdots&\vdots&\vdots\\
0&\cdots&\cdots&0&-\frac{1}{k+1}&-\frac{1}{k+3}&\cdots&-\frac{1}{2k-1}
\end{matrix}\right|+\mathcal{O}(y^{k^2-1}).
\end{split}
\end{align}
Similarly, when $y$ is large, the numerator $\det(\mathbf{G})$ can be expanded as
\begin{equation}
\begin{split}
\det(\mathbf{G})
&=\left|\begin{matrix}-2\ii y&\cdots&\frac{-2\ii y(-\ii y)^{k-2}}{(k-2)!\left(k-1\right)}&0&\cdots&0&1\\
\frac{-2\ii y(-\ii y)^2}{2!3}&\cdots&\frac{-2\ii y(-\ii y)^{k}}{2!(k-2)!(k+1)}&0&\cdots&0&\frac{(-\ii y)^2}{2!}\\
\vdots&\vdots&\vdots&\vdots&\vdots&\vdots&\vdots\\
\frac{-2\ii y(-\ii y)^{k-2}}{(k-2)!(k-1)}&\cdots&\frac{-2\ii y(-\ii y)^{2k-4}}{((k-2)!)^2(2k-3)}&0&\cdots&0&\frac{(-\ii y)^{k-2}}{(k-2)!}\\
0&\cdots&0&\frac{2\ii y(-\ii y)^2}{3}&\cdots&\frac{2\ii y(-\ii y)^{k}}{(k-1)!(k+1)}&-\ii y\\
0&\cdots&0&\frac{2\ii y(-\ii y)^{4}}{3!5}&\cdots&\frac{2\ii y(-\ii y)^{k+2}}{3!(k-1)!(k+3)}&\frac{(-\ii y)^{3}}{3!}\\
\vdots&\vdots&\vdots&\vdots&\vdots&\vdots&\vdots\\
0&\cdots&0&\frac{2\ii y(-\ii y)^{k}}{(k-1)!(k+1)}&\cdots&\frac{2\ii y(-\ii y)^{2k-2}}{((k-1)!)^2(2k-1)}&\frac{(-\ii y)^{k-1}}{(k-1)!}\\
1&\cdots&\frac{(-\ii y)^{k-2}}{(k-2)!}&-\ii y&\cdots&\frac{(-\ii y)^{k-1}}{(k-1)!}&0
\end{matrix}\right|+\mathcal{O}(y^{k^2-2}) \\
\end{split}
\end{equation}
\begin{equation*}
\begin{split}
&=\frac{2^{k}(-\ii y)^{k^2-1}}{\left(2!3!\cdots(k-1)!\right)^2}\left|\begin{matrix}1&\frac{1}{3}&\cdots&\frac{1}{k-1}&0&\cdots&\cdots&0&\frac{1}{2}\\
\frac{1}{3}&\frac{1}{5}&\cdots&\frac{1}{k+1}&0&\cdots&\cdots&0&\frac{1}{2}\\
\vdots&\vdots&\vdots&\vdots&\vdots&\vdots&\vdots&\vdots&\vdots\\
\frac{1}{k-1}&\frac{1}{k+1}&\cdots&\frac{1}{2k-3}&0&\cdots&\cdots&0&\frac{1}{2}\\
0&\cdots&\cdots&0&-\frac{1}{3}&-\frac{1}{5}&\cdots&-\frac{1}{k+1}&\frac{1}{2}\\
0&\cdots&\cdots&0&-\frac{1}{5}&-\frac{1}{7}&\cdots&-\frac{1}{k+3}&\frac{1}{2}\\
\vdots&\vdots&\vdots&\vdots&\vdots&\vdots&\vdots&\vdots&\vdots\\
0&\cdots&\cdots&0&-\frac{1}{k+1}&-\frac{1}{k+3}&\cdots&-\frac{1}{2k-1}&\frac{1}{2}\\
1&\cdots&\cdots&1&1&\cdots&\cdots&1&0
\end{matrix}\right|+\mathcal{O}(y^{k^2-2}).
\end{split}
\end{equation*}
Therefore,
\begin{equation}
p^{[k]}=-\frac{2}{y}\frac{\left|\begin{matrix}\mathbf{H}_{\rm e}&\mathbf{\frac{1}{2}}\\
\mathbf{1}&0
\end{matrix}\right|}{\det(\mathbf{H}_{\rm e})}+\mathcal{O}(y^{-2}),
\end{equation}
where $\mathbf{H}_{\rm e}$ is a block matrix defined by
\begin{equation}
\mathbf{H}_{\rm e}=\begin{bmatrix}1&\frac{1}{3}&\cdots&\frac{1}{k-1}&0&\cdots&\cdots&0\\
\frac{1}{3}&\frac{1}{5}&\cdots&\frac{1}{k+1}&0&\cdots&\cdots&0\\
\vdots&\vdots&\vdots&\vdots&\vdots&\vdots&\vdots&\vdots\\
\frac{1}{k-1}&\frac{1}{k+1}&\cdots&\frac{1}{2k-3}&0&\cdots&\cdots&0\\
0&\cdots&\cdots&0&-\frac{1}{3}&-\frac{1}{5}&\cdots&-\frac{1}{k+1}\\
0&\cdots&\cdots&0&-\frac{1}{5}&-\frac{1}{7}&\cdots&-\frac{1}{k+3}\\
\vdots&\vdots&\vdots&\vdots&\vdots&\vdots&\vdots&\vdots\\
0&\cdots&\cdots&0&-\frac{1}{k+1}&-\frac{1}{k+3}&\cdots&-\frac{1}{2k-1}
\end{bmatrix}.
\end{equation}
Here we use the subscript $_{\rm e}$ to mean the even case.
With a simple calculation, we have
\begin{equation}\label{eq:pm-asy}
p^{[k]}=-\frac{2}{y}\left(-\left(1,1,\cdots, 1\right)\mathbf{H}_{\rm e}^{-1}\left(\frac{1}{2}, \frac{1}{2}, \cdots, \frac{1}{2}\right)^{\T}\right)+\mathcal{O}(y^{-2})=\frac{\sum_{i,j=1}^{k}\left(\mathbf{H}_{\rm e}^{-1}\right)_{i,j}}{y}+\mathcal{O}(y^{-2}).
\end{equation}
The upper left block of $\mathbf{H}_{\rm e}$ is a Hilbert matrix, while the lower right block is the negative of a Hilbert matrix. Employing basic property of Hilbert matrices, we can derive the following identity:
$$\sum_{i,j=1}^{k}\left(\mathbf{H}_{\rm e}^{-1}\right)_{i,j}=-k.$$
Hence $p^{[k]}=-\frac{k}{y}+\mathcal{O}(y^{-2})$. Substituting into Eq. \eqref{eq:P-II} and comparing the coefficient of $y^{0}$, we know that the constant $m_{k}$ must be $-k$. Analogously, when $k$ is odd, we may still expand both the denominator and the numerator. In this case, the denominator and the numerator have the following expansions as $y\to\infty$:
\begin{equation}
\det(\mathbf{M})=\frac{2^{k}(-\ii y)^{k^2}}{\left(2!3!\cdots(k-1)!\right)^2}\left|\begin{matrix}1&\frac{1}{3}&\cdots&\frac{1}{k}&0&\cdots&\cdots&0\\
\frac{1}{3}&\frac{1}{5}&\cdots&\frac{1}{k+2}&0&\cdots&\cdots&0\\
\vdots&\vdots&\vdots&\vdots&\vdots&\vdots&\vdots&\vdots\\
\frac{1}{k}&\frac{1}{k+2}&\cdots&\frac{1}{2k-1}&0&\cdots&\cdots&0\\
0&\cdots&\cdots&0&-\frac{1}{3}&-\frac{1}{5}&\cdots&-\frac{1}{k}\\
0&\cdots&\cdots&0&-\frac{1}{5}&-\frac{1}{7}&\cdots&-\frac{1}{k+2}\\
\vdots&\vdots&\vdots&\vdots&\vdots&\vdots&\vdots&\vdots\\
0&\cdots&\cdots&0&-\frac{1}{k}&-\frac{1}{k+2}&\cdots&-\frac{1}{2k-3}
\end{matrix}\right|+\mathcal{O}(y^{k^2-1}),
\end{equation}
and
\begin{equation}
\det(\mathbf{G})=
\frac{2^{k}(-\ii y)^{k^2-1}}{\left(2!3!\cdots(k-1)!\right)^2}\left|\begin{matrix}1&\frac{1}{3}&\cdots&\frac{1}{k}&0&\cdots&\cdots&0&\frac{1}{2}\\
\frac{1}{3}&\frac{1}{5}&\cdots&\frac{1}{k+2}&0&\cdots&\cdots&0&\frac{1}{2}\\
\vdots&\vdots&\vdots&\vdots&\vdots&\vdots&\vdots&\vdots&\vdots\\
\frac{1}{k}&\frac{1}{k+2}&\cdots&\frac{1}{2k-1}&0&\cdots&\cdots&0&\frac{1}{2}\\
0&\cdots&\cdots&0&-\frac{1}{3}&-\frac{1}{5}&\cdots&-\frac{1}{k}&\frac{1}{2}\\
0&\cdots&\cdots&0&-\frac{1}{5}&-\frac{1}{7}&\cdots&-\frac{1}{k+2}&\frac{1}{2}\\
\vdots&\vdots&\vdots&\vdots&\vdots&\vdots&\vdots&\vdots&\vdots\\
0&\cdots&\cdots&0&-\frac{1}{k}&-\frac{1}{k+2}&\cdots&-\frac{1}{2k-3}&\frac{1}{2}\\
1&\cdots&\cdots&1&1&\cdots&\cdots&1&0
\end{matrix}\right|+\mathcal{O}(y^{k^2-2}).
\end{equation}
A similar asymptotic formula as Eq.\eqref{eq:pm-asy} for $p^{[k]}$ with odd $k$ can also be obtained, with $\mathbf{H}_\mathrm{e}$ replaced by $\mathbf{H}_\mathrm{o}$. The subscript $_{\rm o}$ indicates the odd case. In the odd case, $\mathbf{H}_{\rm o}$ is instead given by
\begin{equation}
\mathbf{H}_{\rm o}=\begin{bmatrix}1&\frac{1}{3}&\cdots&\frac{1}{k}&0&\cdots&\cdots&0\\
\frac{1}{3}&\frac{1}{5}&\cdots&\frac{1}{k+2}&0&\cdots&\cdots&0\\
\vdots&\vdots&\vdots&\vdots&\vdots&\vdots&\vdots&\vdots\\
\frac{1}{k}&\frac{1}{k+2}&\cdots&\frac{1}{2k-1}&0&\cdots&\cdots&0\\
0&\cdots&\cdots&0&-\frac{1}{3}&-\frac{1}{5}&\cdots&-\frac{1}{k}\\
0&\cdots&\cdots&0&-\frac{1}{5}&-\frac{1}{7}&\cdots&-\frac{1}{k+2}\\
\vdots&\vdots&\vdots&\vdots&\vdots&\vdots&\vdots&\vdots\\
0&\cdots&\cdots&0&-\frac{1}{k}&-\frac{1}{k+2}&\cdots&-\frac{1}{2k-3}
\end{bmatrix}.
\end{equation}
For this matrix, we have $$\sum_{i,j=1}^{k}\left(\mathbf{H}_{\rm o}^{-1}\right)_{i,j}=k.$$
Consequently, the asymptotics of $p^{[k]}$ becomes $p^{[k]}=\frac{k}{y}+\mathcal{O}(y^{-2})$. Substituting this expression into Eq.\eqref{eq:P-II}, then we obtain that the constant $m_k$ is equal to $k$.
\end{proof}

\noindent{\bf{Example:}}

Next, we compute here the first several rational solutions derived from the generalized Darboux transformation. When $k=1$, the Darboux transformation $\mathbf{T}_{1}(\lambda; y)$ and the first order solution $p^{[1]}$ are
\begin{equation}
\mathbf{T}_{1}(\lambda; y)=\mathbb{I}-\frac{\ii}{2\lambda y}\begin{bmatrix}-1&-1\\1&1
\end{bmatrix},\quad p^{[1]}=\frac{1}{y}.
\end{equation}
When $k=2$, the Darboux transformation $\mathbf{T}^{[2]}(\lambda; y)$ and the second order rational solution $p^{[2]}$ are
\begin{equation}
\mathbf{T}^{[2]}(\lambda; y)=\mathbb{I}+\begin{bmatrix}\frac{-4\ii\lambda y^3-6\ii\lambda-3y^2}{2\lambda^2y\left(y^3+6\right)}&\frac{2\ii\lambda y^3-6\ii\lambda-3y^2}{2\lambda^2y\left(y^3+6\right)}\\-\frac{2\ii\lambda y^3-6\ii\lambda+3y^2}{2\lambda^2y\left(y^3+6\right)}&\frac{4\ii\lambda y^3+6\ii\lambda+3y^2}{2\lambda^2y\left(y^3+6\right)}
\end{bmatrix},\quad p^{[2]}=-\frac{2\left(y^3-3\right)}{y\left(y^3+6\right)}.
\end{equation}

Obviously, the first order rational solution $p^{[1]}$ satisfies the Painlev\'{e}-II equation \eqref{eq:P-II} with $m_1=1$ and the second order rational solution $p^{[2]}$ satisfies it with $m_2=-2$. Similar to the rational solutions of KdV equation \cite{Ablowitz-JMP-1978}, both the rational solutions $p^{[1]}$ and $p^{[2]}$ are meromorphic functions with simple poles. Indeed these solutions have different analytic properties from the rogue waves, even though in \cite{Yang-phyd-2021-1} the authors showed the close connection of the two, which is an important motivation for our study. In Fig.\ref{p-II}, we show the location of zeros and poles of higher order rational solutions.
\begin{figure}[!h]
\centering
\includegraphics[width=0.3\textwidth]{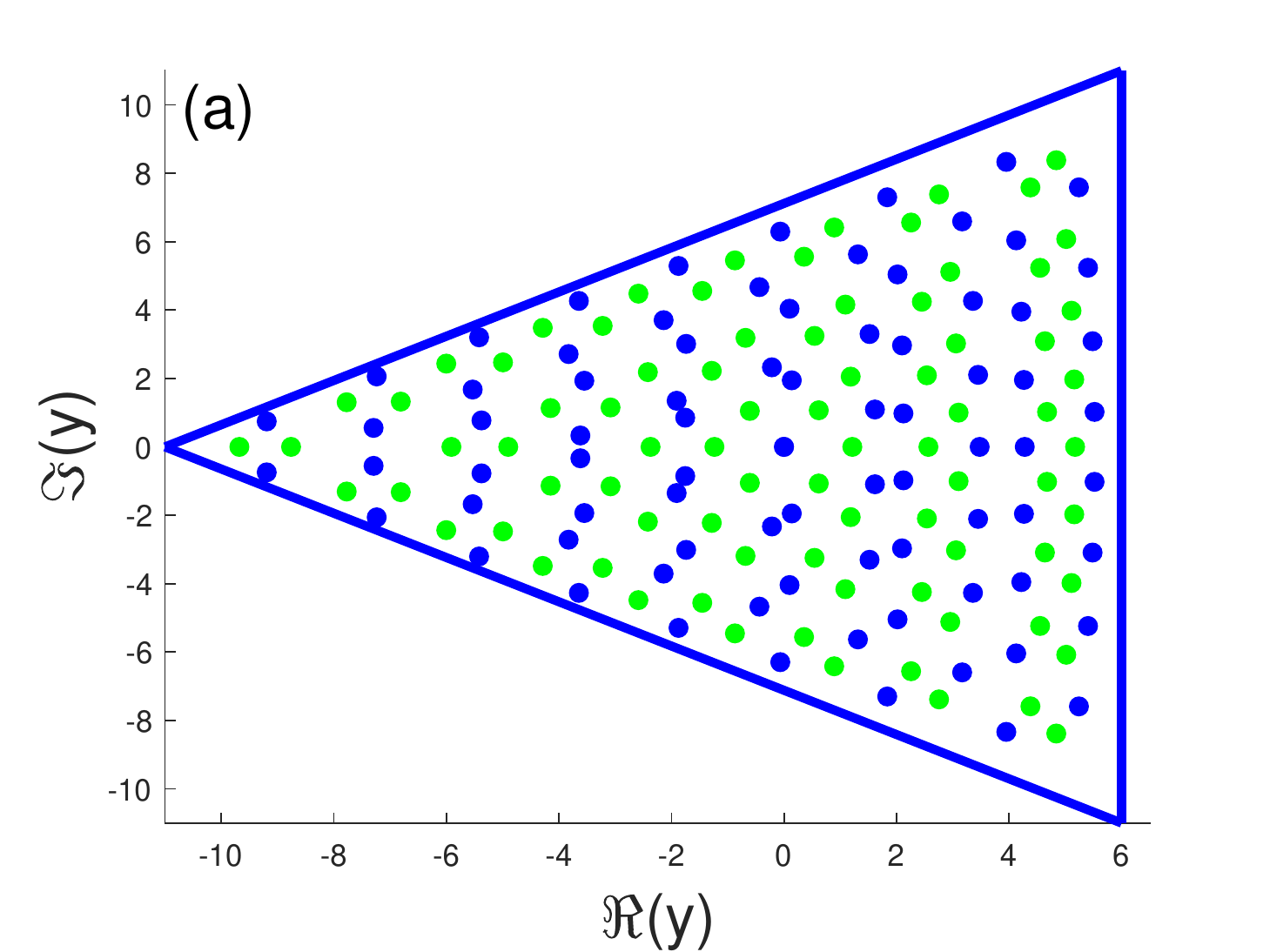}
\centering
\includegraphics[width=0.3\textwidth]{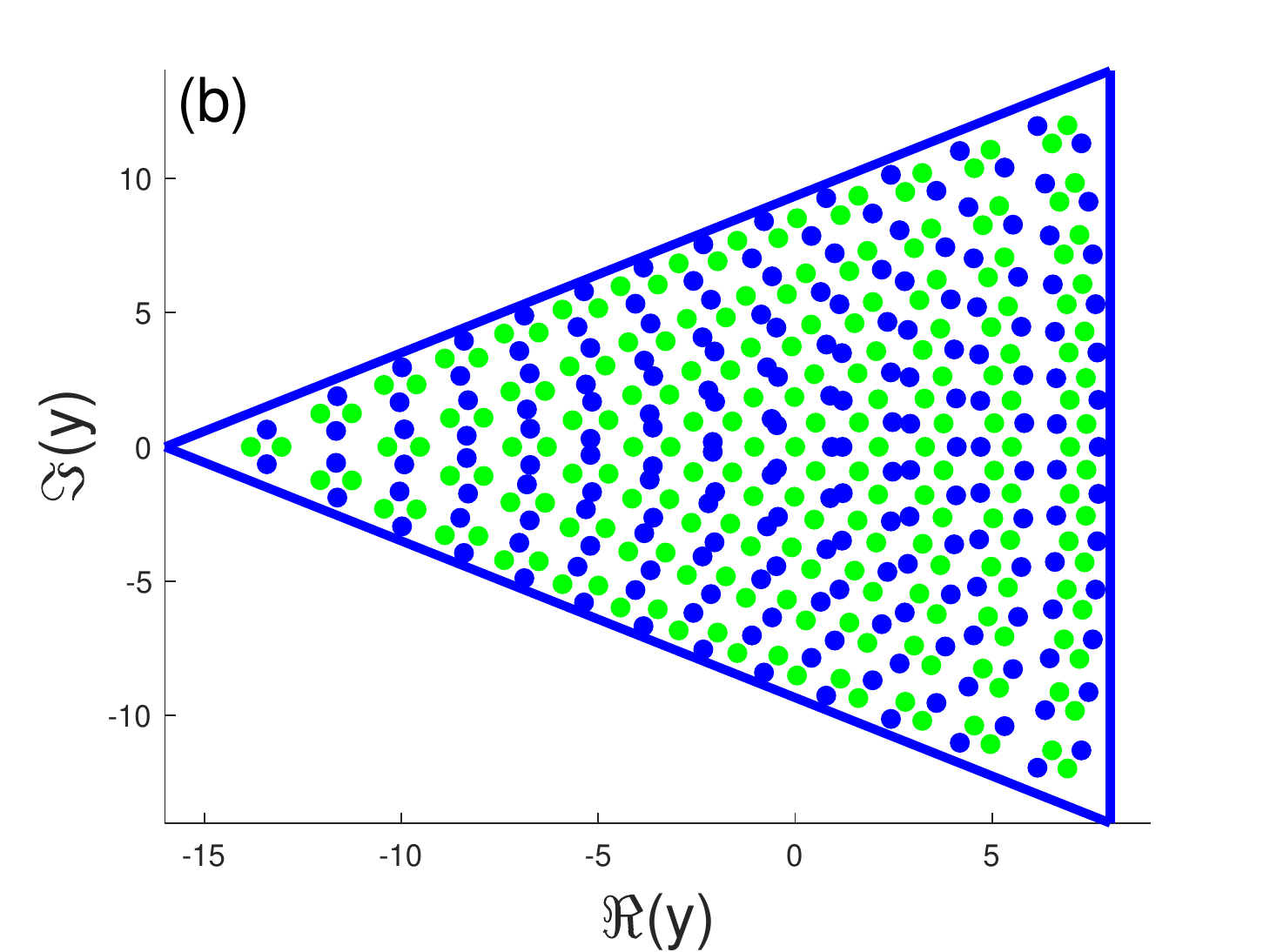}
\centering
\includegraphics[width=0.3\textwidth]{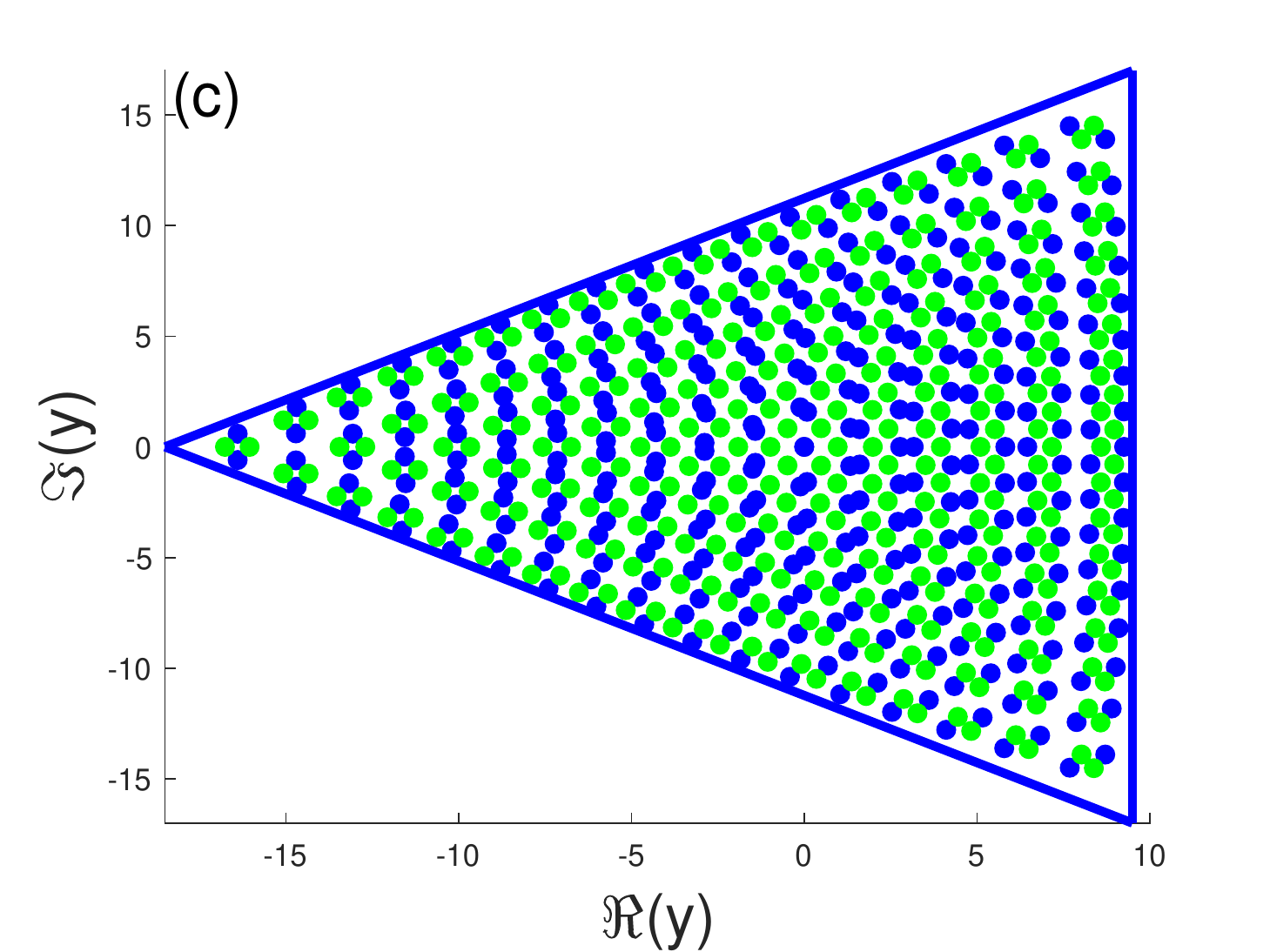}
\caption{The poles (Green) and the zeros (Blue) of the rational solutions of Painlev\'{e}-II equation. The corresponding parameters are $m_9=9$  (left), $m_{14}=-14$ (middle) and $m_{18}=-18$ (right), respectively.}
\label{p-II}
\end{figure}

\section{Riemann-Hilbert representation about the rational solutions of Painlev\'{e}-II equation}\label{sec:RH}
In the previous two sections, we constructed two types of Darboux transformations about the rational solutions and derived the corresponding B\"{a}cklund transformation. Both of them are studied under the Flaschka-Newell Lax representation Eq.\eqref{eq:Lax-pair}. It is well known that the Painlev\'{e}-II equation has three different Lax representations, i.e. the Flaschka-Newell Lax representation, the Jimbo-Miwa Lax representation and Bortola-Bothner representation. In \cite{Peter-Sigma-2017}, the authors gave a detailed introduction of the three representations and constructed three different Riemann-Hilbert problems accordingly. In this section, we will give the solution of the first (Flaschka-Newell) Riemann-Hilbert problem for Painlev\'e-II equation.

First we give a brief review of the Riemann-Hilbert problem corresponding to the Flaschka-Newell Lax representation from \cite{Peter-Sigma-2017}. Firstly, the authors constructed two types of fundamental solution matrices $\mathbf{V}_{\infty}(\lambda; y)$ and $\mathbf{V}_{0}(\lambda; y)$ of the Flaschka-Newell Lax pair Eq.\eqref{eq:Lax-pair}. The former solution $\mathbf{V}_{\infty}(\lambda; y)$ has a convergent expansion when $|\lambda|$ is large and the latter solution is given in the neighbourhood of $\lambda=0$. Then there exists a so-called monodromy matrix $\mathbf{G}_{m}(\lambda; y)$ such that \begin{equation}
\mathbf{V}_{\infty}(\lambda; y)=\mathbf{V}_{0}(\lambda; y)\mathbf{G}_{m}, \quad 0<|\lambda|<\infty.
\end{equation}
One can show the following symmetry of $\mathbf{V}_{\infty}(\lambda; y)$ and $\mathbf{V}_{0}(\lambda; y)$
\begin{equation}
\sigma_1\mathbf{V}_{\infty}(-\lambda; y)\sigma_1=\mathbf{V}_{\infty}(\lambda; y), \quad \sigma_1\mathbf{V}_{0}(-\lambda; y)\sigma_1=\mathbf{V}_{0}(\lambda; y)\begin{bmatrix}0&(-1)^m\\
(-1)^{m+1}&0
\end{bmatrix}.
\end{equation}
The general formula of the connection matrix $\mathbf{G}_{m}$ is given by $\mathbf{G}_{m}=\begin{bmatrix}\alpha&(-1)^m\alpha\\
(-1)^{m+1}(2\alpha)^{-1}&(2\alpha)^{-1}
\end{bmatrix}$, $\alpha\neq 0$. Define two sectional matrices $\mathbf{M}_{m}(\lambda; y)$ depending on these two fundamental solutions $\mathbf{V}_{\infty}(\lambda; y)$ and $\mathbf{V}_{0}(\lambda; y)$, that is
\begin{equation}
\mathbf{M}_{m}(\lambda; y)=\left\{
\begin{split}\mathbf{V}_{\infty}(\lambda; y)\ee^{\ii\theta(\lambda; y)\sigma_3}\lambda^{-m\sigma_3},\quad |\lambda|>1,\\
\mathbf{V}_{0}(\lambda; y)\ee^{\ii\theta(\lambda; y)\sigma_3}\lambda^{-m\sigma_3},\quad |\lambda|<1,
\end{split}
\right.
\end{equation}
where $\theta(\lambda; y)=\lambda y+2\lambda^3$. This immediately gives a Riemann-Hilbert problem for $\mathbf{M}_{m}(\lambda; y)$.
\begin{rhp}[\cite{Peter-Sigma-2017}]\label{rhp:1}
Let $m\in\mathbb{Z}$ and $y\in \mathbb{C}$, seek $\mathbf{M}_{m}(\lambda; y)$ satisfying the following conditions.
\begin{itemize}
\item \textbf{Analyticity:} $\mathbf{M}_{m}(\lambda; y)$ is analytic for $|\lambda|\neq 1$, and $\mathbf{M}_{m,+}(\lambda; y)$, $\mathbf{M}_{m,-}(\lambda; y)$ are the continuous boundary values from the interior and the exterior.
\item \textbf{Jump condition:} When $|\lambda|=1$, $\mathbf{M}_{m}(\lambda; y)$ satisfies the following jump condition,
\begin{equation}
    \mathbf{M}_{m,+}(\lambda; y)=\mathbf{M}_{m,-}(\lambda;y)\lambda^{m\sigma_3}\ee^{-\ii\theta(\lambda; y)\sigma_3}\mathbf{G}^{-1}_{m}\ee^{\ii\theta(\lambda; y)\sigma_3}\lambda^{-m\sigma_3}, \quad |\lambda|=1.
\end{equation}
\item \textbf{Normalization:} When $\lambda\to\infty$, $\mathbf{M}_{m}(\lambda; y)$ is normalized with the following formula,
\begin{equation}
\lim\limits_{\lambda\to\infty}\mathbf{M}_{m}(\lambda; y)\lambda^{m\sigma_3}=\mathbb{I}.
\end{equation}
\end{itemize}
\end{rhp}
Then the rational solutions are given by
\begin{equation}
p^{[m]}(y)=2\ii\lim\limits_{\lambda\to\infty}\lambda^{1+m}M_{m,12}(\lambda; y)=-2\ii\lim\limits_{\lambda\to\infty}\lambda^{1-m}M_{m,21}(\lambda; y).
\end{equation}

To solve this Riemann-Hilbert problem, in section \ref{sec:DT-1} and \ref{sec:DT-2}, we constructed two types of generalized Darboux transformation and derived the general rational solutions of Painlev\'{e}-II equation. The solution of the Riemann-Hilbert problem \ref{rhp:1} can be directly given by the Darboux transformation.
\begin{theorem}\label{theo-3}
The solution of Riemann-Hilbert problem \ref{rhp:1} can be given by
\begin{equation}
\mathbf{M}_{m}(\lambda; y)=\left\{\begin{aligned}
\mathbf{M}_{m,-}(\lambda;y)=&(-1)^{\frac{m+1}{2}\sigma_3}\mathbf{T}^{[m]}(\lambda; y)(-1)^{\frac{m+1}{2}\sigma_3}\lambda^{-m\sigma_3}, \\
\mathbf{M}_{m,+}(\lambda;y)=&(-1)^{\frac{m+1}{2}\sigma_3}\mathbf{T}^{[m]}(\lambda; y)(-1)^{\frac{m+1}{2}\sigma_3}\ee^{-\ii\theta\sigma_3}\mathbf{G}_{m}^{-1}\ee^{\ii\theta\sigma_3}\lambda^{-m\sigma_3},
\end{aligned}\right.
\end{equation}
where $\mathbf{T}^{[m]}(\lambda; y)$ is the generalized Darboux transformation defined in \eqref{eq:darboux}.
\end{theorem}
\begin{proof}
Based on the basic property of Darboux transformation, we know that $\mathbf{T}^{[m]}(\lambda; y)\ee^{-\ii\theta\sigma_3}$ can solve the Lax pair Eq.\eqref{eq:Lax-pair-1} with $p$ taken as the $m$th order rational solution $p^{[m]}$ and the constant as  $(-1)^{m-1}m$. From the symmetry in the $\partial y$-equation of the Lax pair, we have $\sigma_1\mathbf{B}^{FN}(-\lambda; y)\sigma_1=\mathbf{B}^{FN}(\lambda; y)$. Therefore, the Darboux transformation matrix also have symmetry
$$\sigma_1\mathbf{T}^{[m]}(-\lambda; y)\sigma_1=\mathbf{T}^{[m]}(\lambda; y).$$
In addition, the Darboux transformation matrix has a convergent series expansion for $\lambda\to\infty$. Notice that $\mathbf{T}^{[m]}$ almost coincides with the fundamental solution $\mathbf{V}_{\infty}(\lambda; y)$ defined in \cite{Peter-Sigma-2017}, except for a sign difference for even $m$. With the discrete symmetry of the Painlev\'{e}-II equation $\left(p(y), m\right)\rightarrow \left(-p(y), -m\right)$, we can take a simple transformation to $\mathbf{T}^{[m]}(\lambda; y)$, and now the formula is valid for all $m\in\mathbb{Z}$. Thus, we have derived the connection between the generalized Darboux transformation and the Riemann-Hilbert problem \ref{rhp:1} when $|\lambda|>1$. For $|\lambda|<1$, the solution of the Riemann Hilbert problem \ref{rhp:1} can also be given by the jump condition. Therefore, the solution of the Riemann-Hilbert problem \ref{rhp:1} is given in the whole complex plane $\mathbb{C}$, that is,
\begin{equation}
\begin{aligned}
\mathbf{M}_{m,-}(\lambda;y)=&(-1)^{\frac{m+1}{2}\sigma_3}\mathbf{T}^{[m]}(\lambda; y)(-1)^{\frac{m+1}{2}\sigma_3}\lambda^{-m\sigma_3}, \\
\mathbf{M}_{m,+}(\lambda;y)=&(-1)^{\frac{m+1}{2}\sigma_3}\mathbf{T}^{[m]}(\lambda; y)(-1)^{\frac{m+1}{2}\sigma_3}\ee^{-\ii\theta\sigma_3}\mathbf{G}_{m}^{-1}\ee^{\ii\theta\sigma_3}\lambda^{-m\sigma_3}.
\end{aligned}
\end{equation}
\end{proof}

\section{Discussions and Conclusions}
In this paper, we construct the general rational solutions to Painlev\'{e}-II equation by the generalized Darboux transformation. We are able to compactly write the rational solutions as Gram determinant. In the generalized Darboux transformation, the spectral parameter $\lambda_1$ is chosen as a special value $\lambda_1=0$. Under this condition, the fundamental solution $\phi_{1}$ becomes a polynomial. In particular, the spectrum in the Darboux transformation is a removable singularity, which is similar to the study of rogue wave. In section \ref{sec:DT-1}, we use two methods to derive the  B\"{a}cklund transformation of the rational solutions. One is directly from the Painlev\'{e}-II equation itself, and the other one is derived from the iteration steps in the Darboux transformation. The two transformations are shown to be equivalent by a simple transformation. In addition, from the exact form of rational solutions simply represented by Gram determinants, we can also analyze the asymptotics of large $y$. In section \ref{sec:RH}, we prove that our Darboux transformation can solve the Riemann-Hilbert problem in \cite{Peter-Sigma-2017}, which has never been reported before.

It is known that the generalized Darboux transformation can be used to derive the rogue wave of NLS equation and the rational solutions of Painlev\'{e}-II equation, both written as a Gram determinant. In view of the result in \cite{Yang-phyd-2021}, we conjecture that there exist certain connections between these two kinds of solutions. Our goal is to continue the study of this topic via the Riemann-Hilbert problem in the future.

\section*{Acknowledgements}

Liming Ling is supported by the National Natural Science Foundation of China (Grant No. 12122105), the Guangzhou Science and Technology Program of China (Grant No. 201904010362); Xiaoen Zhang is supported by the National Natural Science Foundation of China (Grant No.12101246), the China Postdoctoral Science Foundation (Grant No. 2020M682692), the Guangzhou Science and Technology Program of China(Grant No. 202102020783).

\end{document}